\documentclass[journal]{IEEEtran}
\ifCLASSINFOpdf
\else
\fi

\usepackage{cite}
\usepackage{amsmath}
\usepackage{bbm}
\usepackage{amsfonts}
\usepackage{amsthm}
\usepackage{mathtools}
\usepackage{graphicx}
\usepackage{hyperref}
\usepackage{cleveref}
\usepackage{comment}
\graphicspath{ {./images/} }
\usepackage{ifthen}
\newtheorem{lemma}{Lemma}

\newtheorem{theorem}{Theorem}
\usepackage{nicefrac}

\usepackage{amsmath, amssymb, bbm, xspace}
\usepackage{epsfig}
\usepackage{longtable}
\usepackage{color}
\usepackage{mathrsfs}
\usepackage{comment}

\usepackage{courier}

\def\bkE{{\rm I\kern-.17em E}}
\def\bk1{{\rm 1\kern-.17em l}}
\def\bkD{{\rm I\kern-.17em D}}
\def\bkR{{\rm I\kern-.17em R}}
\def\bkP{{\rm I\kern-.17em P}}

\def\bkZ{{\bf{Z}}}

\def\bkE{{\rm I\kern-.17em E}}
\def\bk1{{\rm 1\kern-.17em l}}
\def\bkD{{\rm I\kern-.17em D}}
\def\bkR{{\rm I\kern-.17em R}}
\def\bkP{{\rm I\kern-.17em P}}

\makeatletter
\newcommand{\pushright}[1]{\ifmeasuring@#1\else\omit\hfill$\displaystyle#1$\fi\ignorespaces}
\newcommand{\pushleft}[1]{\ifmeasuring@#1\else\omit$\displaystyle#1$\hfill\fi\ignorespaces}
\makeatother

\def\bkZ{{\bf{Z}}}
\def\b12{(\beta_1,\beta_2)}

\newenvironment{example}{{\noindent \bf Example}}{\hfill $\square$\hspace{-4.5pt}\vspace{6pt}}
\newcounter{example}
\renewcommand{\theexample}{\thesection.\arabic{example}}

\newcounter{remark}
\renewcommand{\theremark}{\thesection.\arabic{remark}}

\def\t{^\top}

\newlength{\noteWidth}
\long\def\notes#1{\ifinner
{\tiny #1}
\else
\marginpar{\parbox[t]{\noteWidth}{\raggedright\tiny #1}}
\fi\typeout{#1}}

 \def\notes#1{\typeout{read notes: #1}} 

\newcommand{\inv}{^{-1}}

\def\ast{\buildrel{t}\over\rightarrow}

\def\diag{\mathop{\hbox{\rm diag}}}
  
\def\exp{\mathop{\hbox{\rm exp}}}

\def\spose#1{\hbox to 0pt{#1\hss}}

\def\text #1{\hbox{\quad#1\quad}}

\def\nthinsp{\mskip -2   mu}

\def\superstar{^{\raise 0.5pt\hbox{$\nthinsp *$}}}
\def\SUPERSTAR{^{\raise 0.5pt\hbox{$*$}}}

\def\lamstarT {\lambda^{\raise 0.5pt\hbox{$\nthinsp *$}T}}

\let\forallnew\forall
\renewcommand{\forall}{\forallnew\ }
\let\forall\forallnew

		\def\bkE{{\rm I\kern-.17em E}}
		\def\bk1{{\rm 1\kern-.17em l}}
		\def\bkD{{\rm I\kern-.17em D}}
		\def\bkR{{\rm I\kern-.17em R}}
		\def\bkP{{\rm I\kern-.17em P}}
		\def\bkY{{\bf \kern-.17em Y}}
		\def\bkZ{{\bf \kern-.17em Z}}
		\def\bkC{{\bf  \kern-.17em C}}

{\begin{list}{}%
         {\setlength{\leftmargin}{#1}}%
         \item[]%
}
{\end{list}}

		\def\bsp{\begin{split}}
		\def\beq{\begin{eqnarray}}
		\def\bal{\begin{align*}}
		\def\bc{\begin{center}}
		\def\be{\begin{enumerate}}
		\def\bi{\begin{itemize}}
		\def\bs{\begin{small}}
		\def\bS{\begin{slide}}
		\def\ec{\end{center}}
		\def\ee{\end{enumerate}}
		\def\ei{\end{itemize}}
		\def\es{\end{small}}
		\def\eS{\end{slide}}
		\def\eeq{\end{eqnarray}}
		\def\eal{\end{align*}}
		\def\esp{\end{split}}
		\def\qed{ \vrule height7.5pt width7.5pt depth0pt}  

	\def\cp2problem#1#2#3#4{\fbox
		 {\begin{tabular*}{0.9\textwidth}
			{@{}l@{\extracolsep{\fill}}l@{\extracolsep{6pt}}l@{\extracolsep{\fill}}c@{}}
				#1 & & $#4 $ 
			\end{tabular*}}}

		\def\bkE{{\rm I\kern-.17em E}}
		\def\bk1{{\rm 1\kern-.17em l}}
		\def\bkD{{\rm I\kern-.17em D}}
		\def\bkR{{\rm I\kern-.17em R}}
		\def\bkP{{\rm I\kern-.17em P}}
		
		\def\bkZ{{\bf{Z}}}

\newcommand {\beeq}[1]{\begin{equation}\label{#1}}
\newcommand {\eeeq}{\end{equation}}
\newcommand {\bea}{\begin{eqnarray}}
\newcommand {\eea}{\end{eqnarray}}

\def\texitem#1{\par\smallskip\noindent\hangindent 25pt
               \hbox to 25pt {\hss #1 ~}\ignorespaces}

\def\bsp{\begin{split}}
		\def\beq{\begin{eqnarray}}
		\def\bal{\begin{align*}}
		\def\bc{\begin{center}}
		\def\be{\begin{enumerate}}
		\def\bi{\begin{itemize}}
		\def\bs{\begin{small}}
		\def\bS{\begin{slide}}
		\def\ec{\end{center}}
		\def\ee{\end{enumerate}}
		\def\ei{\end{itemize}}
		\def\es{\end{small}}
		\def\eS{\end{slide}}
		\def\eeq{\end{eqnarray}}
		\def\eal{\end{align*}}
		\def\esp{\end{split}}
		\def\qed{ \vrule height7.5pt width7.5pt depth0pt}  

\usepackage{amsmath, amssymb, xspace}
\usepackage{epsfig}
\usepackage{longtable}
\usepackage{color}
\usepackage{mathrsfs}
\usepackage{subfig}

\def\LOLP{\mathsf{LOLP}}

\newboolean{showcomments}
\setboolean{showcomments}{true}
\newcommand{\jk}[1]{  \ifthenelse{\boolean{showcomments}}
{ \textcolor{red}{(JK says:  #1)}} {}  }
\newcommand{\aak}[1]{  \ifthenelse{\boolean{showcomments}}
{ \textcolor{blue}{(AAK says:  #1)}} {}  }
\newcommand{\vivek}[1]{  \ifthenelse{\boolean{showcomments}}
{ \textcolor{red}{(Vivek says:  #1)}} {}  }

\newcommand{\ra}{\rightarrow}
\newcommand{\ignore}[1]{}
\newcommand{\ua}{\uparrow}
\newcommand{\da}{\downarrow}
\newcommand{\prob}[1]{\mathbb{P}\left[#1\right]}
 
\newcommand{\Exp}[1]{\mathbb{E}\left[#1\right]} 

\newcommand{\indicator}{\boldsymbol{1}}

\begin{document}

\title{Sizing Storage for Reliable Renewable Integration: A Large
  Deviations Approach}
\author{Vivek Deulkar, \ Jayakrishnan Nair and Ankur A. Kulkarni
\thanks{Vivek and Jayakrishnan are with the Department of Electrical Engineering, 
Indian Institute of Technology Bombay, Mumbai, India. Ankur is with the Systems and 
Control Engineering group, Indian Institute of Technology Bombay, Mumbai, India. 
The authors can be reached at \texttt{\footnotesize vivekdeulkar@iitb.ac.in, jayakrishnan.nair@iitb.ac.in, 
kulkarni.ankur@iitb.ac.in}. A preliminary version of this paper will be presented at IEEE PowerTech 2019\cite{deulkar2019sizing}.}}  

\maketitle

\begin{abstract}
  The inherent intermittency of wind and solar generation presents a
  significant challenge as we seek to increase the penetration of
  renewable generation in the power grid. Increasingly, energy storage
  is being deployed alongside renewable generation to counter this
  intermittency. However, a formal characterization of the reliability
  of renewable generators bundled with storage is lacking in the
  literature. The present paper seeks to fill this gap. We use a
  Markov modulated fluid queue to model the loss of load probability
  ($\LOLP$) associated with a renewable generator bundled with a
  battery, serving an uncertain demand process. Further, we
  characterize the asymptotic behavior of the $\LOLP$ as the battery
  size scales to infinity. Our results shed light on the fundamental
  limits of reliability achievable, and also guide the sizing of the
  storage required in order to meet a given reliability
  target. Finally, we present a case study using real-world wind power
  data to demonstrate the applicability of our results in practice.
\end{abstract}

\IEEEpeerreviewmaketitle

\section{Introduction}

\IEEEPARstart{E}lectric supply is an indispensable part of modern  life and is thus required to meet extremely stringent requirements of reliability. Classically, loss of load has been caused due to operational reasons, such as a generator undergoing maintenance, grid conditions, such as the overdrawing of power, or due to extraneous circumstances, such as natural calamities. 
With increasing penetration of renewable generation, the natural variability of the output of these generators adds a new, supply-side cause for the loss of load. Fortunately, with the growing capacity of renewable generation, we are also witnessing a softening of storage prices. Thanks to this, an increasing number of renewable generators are countering their variability, not with conventional, fast-ramping generation, but rather with storage \cite{url_energy_storage_association, url_US_battery_storage_market_trend}. Thus, we believe that the renewable generator of the future  will not be a standalone renewable generator, but rather a renewable  generator \textit{bundled} with a battery. 

Keeping in mind reliability as one of the central concerns of the electricity infrastructure, the introduction of a battery-renewable generator bundle raises some basic questions. To begin, how does one account for this bundle in calculations for system reliability? How does this reliability change with increasing variability of the renewable source (wind or solar)? How does this change with increasing capacity of the battery? If one targets a certain level of reliability, how much battery storage is required to attain this level? And finally, are there fundamental limitations on the performance of a bundle, in the sense that are there levels of performance that are simply unattainable?

A moment's thought reveals that answers to these questions cannot be
obtained by only considering one snapshot in time. To understand this,
consider the hypothetical scenario where there is no battery and only
a renewable generator attached to a constant load. Then the loss of
load probability ($\LOLP$) would be simply the probability that the
instantaneous output of the generator drops below the load, which one
could potentially calculate via meteorological data. However,
introducing storage changes the picture dramatically. Even while the
instantaneous output of the generator may drop to a low level, there
may well be charge left in the battery to meet the load requirements,
and thus, using the bundle, the load could still be met. But the
battery is charged by the excess output of the renewable generator,
whereby the charge in the battery at any time depends on the history
of generation (and load) realized until that time. It is easy to see
that characterization of the $\LOLP$ in this case is a nontrivial
matter.

This paper develops an analytical framework for characterizing the $\LOLP$ of a battery-renewable generator bundle. Our framework yields crisp answers to the sizing questions raised above. 
For a target level of reliability, it provides order-optimal estimates of the minimum battery size one requires to meet that reliability level, in terms of the statistical properties of the renewable source and the load. It also reveals that there are hard impossibilities: for certain ranges of these statistical parameters, no amount of battery suffices to bring the $\LOLP$ to zero. These results could be applied in conjunction with a costing exercise to ascertain the right battery size to be bundled with a renewable generator. One could also potentially use our characterization of the steady-state $\LOLP$ within a larger calculation of network-level reliability.

We model the \textit{net generation}, i.e., the renewable generation
minus the demand, as a continuous time Markov chain evolving over a
finite state space. The battery serves as a buffer of finite capacity
that is charged at the available rate when the net generation is
positive and is discharged at the deficit rate when the net generation
is negative. The battery charging process is subject to `boundary
conditions': it cannot be charged above its capacity and cannot be
discharged below zero. Any positive net generation produced when the
battery is fully charged is unusable. The $\LOLP$ is then the long run
fraction of time the battery is discharged to zero.
We find that when the \textit{drift}, i.e., the steady-state average
net generation, is negative, then a battery of any finite size results
in an $\LOLP$ that remains bounded away from zero. In other words, the
$\LOLP$ cannot be made arbitrarily small by choosing a large enough
battery size when the drift is negative.
However, when the drift is positive, the $\LOLP$ drops exponentially
with battery size, allowing it to be made arbitrarily close to zero by
choosing a suitably large battery size. The rate of decrease of
$\LOLP$ with increase in battery size is dictated by a large
deviations decay rate, which can also be characterized as the smallest
positive generalized eigenvalue of the rate matrix associated with the
net generation. This decay rate characterization can in turn be used
to estimate the battery size required to achieve a given target
$\LOLP.$

This paper is organized as follows. In Section~\ref{sec:model}, we
develop the mathematical model for the renewable source, load and the
battery. In Section~\ref{sec:bounds}, we characterize the asymptotics
of the $\LOLP$ as the battery size scales to infinity. This serves as
the basis of our sizing estimates. We then do a case study in
Section~\ref{case study} where these results are tested and validated
on real data of wind generation.

\section{Model and Preliminaries}\label{sec:model}


Consider a storage battery of capacity $b_{\max}$ which is charged or
discharged by a \textit{net generation} process associated with rate
$r(t)=g(t)-d(t)$, where $g(t)$ and $d(t)$ denote, respectively, the
generation and demand at time $t$. The energy content of the battery,
denoted by $b(t),$ evolves as a regulated process having upper cap
$b_{\max}$ and lower cap $0$. Thus, $b(t)$ evolves as
\begin{equation}
\frac{d}{dt}b(t) = 
    \begin{cases}
      
      0           & \text{if }b(t) = 0 \text{and} r(t)<0,\\ 
      0           & \text{if }b(t) = b_{\max} \text{and} r(t)>0,\\
      r(t)    & \text{ otherwise. }
    \end{cases} \label{eq:bdot}
\end{equation}
Note that a fully charged battery cannot be charged further with a rate
$r(t)>0$. Similarly an empty battery cannot be discharged further with a rate
$r(t)<0$. Excluding these two boundary cases, the rate of change of
the battery level is governed by the net generation rate~$r(t)$.  We
assume that the rate $r(t)$
is dependent on the state of a background Markov process,
which collectively captures supply (generation) side variability as
well as demand side variability.

Let $\{X(t)\}$ denote the background Markov process. We assume that
$\{X(t)\}_{t\geq 0}$ is an irreducible, time-reversible,
continuous-time Markov chain (CTMC) over a finite state space $S$.
For every state $i\in S$, we associate a \textit{net generation rate}
$r_i\in \mathbb{R}\setminus\{0\}$ with which the battery is
charged or discharged. Thus, $r(t):=r_{X(t)},$ i.e., the rate of
charging/discharging of the battery is a function of the state of the
background CTMC $\{X(t)\}_{t\geq 0}.$
It is easy to see now that $\{(b(t),X(t))\}$ is a Markov process that
evolves over the state space $[0,b_{\max}]\times S$. Note that this
model also captures charge/discharge rate constraints on the battery;
these would simply be reflected in the range of values taken by the
net generation rates $r_i.$

The above mathematical model, wherein the occupancy of a buffer (or
battery) is modulated by a background Markov process, is referred to
in the queueing literature as a Markov Modulated Fluid Queue (MMFQ);
see \cite{anick1982stochastic, mitra1988stochastic}. In this paper, we
use a finite-buffer MMFQ model to analyse the reliability of a
renewable generator bundled with a battery.

Next, we describe how to characterize the invariant distribution of
the Markov process $\{(b(t),X(t))\},$ which then leads to a
characterization of the loss of load probability ($\LOLP).$ Note that
we are assuming that the process $X(t)$ has no state~$i$ where the net
generation rate is zero.
This allows us to partition the state space $S$ as follows:~$S=S_+\cup
S_-$, where
\begin{equation*}
    S_+=\{i\in S:r_i>0\}, \quad S_-=\{i\in S:r_i<0\}.
\end{equation*} 
We assume that both $S_+$ and $S_-$ are non-empty.\footnote{Indeed,
  if either $S_+$ or $S_-$ is empty, then the battery would forever
  remain completely charged or completely discharged.}

Let $(b,X)$ denote the steady state of the Markov process
$\{(b(t),X(t))\}.$ We capture the invariant distribution of this
process as follows:
\begin{equation*}
  F_i(x)=\mathbb{P}[b\leq x,X=i]\quad \forall\ i \in
  S,\ x\in[0,b_{\max}].
\end{equation*}
The invariant distribution is governed by the ODE 
\begin{equation}
  \frac{d}{dx}{F}(x) = R^{-1}Q\t F(x), 
  \label{eq:fdot} 
\end{equation}
where $F(\cdot)=[F_1(\cdot), F_2(\cdot),\hdots, F_{|S|}(\cdot)]\t,$
$Q$ denotes the transition rate matrix associated with the CTMC
$\{X(t)\},$ and $R:=\diag(r_1,r_2,\hdots,r_{|S|})$ (see
\cite{anick1982stochastic, mitra1988stochastic}).\footnote{Since $r_i
  \neq 0$ for all $i \in S,$ $R^{-1}$ exists.}
The invariant distribution can now be computed using the following
boundary conditions:
\begin{equation}
    F_i(0)=0\ \forall\ i\in S_+;\ F_i(b_{\max})=\pi_i\ \forall\ i\in
    S_-\label{eq:boundary_conditions},
\end{equation}
where $\pi = (\pi_i,i \in S)$ denotes the invariant distribution of
the CTMC $\{X(t)\}.$

The probability that the battery content is less than or equal to $x$
in steady state is given by $\sum_{i\in S} F_i(x)$.  This probability
is of particular relevance for $x=0$. Indeed, the quantity $\sum_{i\in
  S} F_i(0) = \sum_{i\in S_-} F_i(0)$ is the long run fraction of time
the battery is empty, and is also the long run fraction of time that
the demand remains unfulfilled. In other words, this is the loss of
load probability ($\LOLP$), i.e.,
\begin{equation*}
    \LOLP=\sum_{i\in S_-} F_i(0) \label{eq:LOLP_eqn}.
\end{equation*}

The $\LOLP,$ which can only be expressed in closed form for very
simple cases (see below), can be computed numerically by solving the
ODE~\eqref{eq:fdot} using the boundary
conditions~\eqref{eq:boundary_conditions}. However, this computation
does not provide insights into the structural dependence of the
$\LOLP$ on the supply-side and demand-side uncertainty (captured by the
CTMC $\{X(t)\}$) and the capacity $b_{\max}$ of the battery. In
Section~\ref{sec:bounds}, we analyse the large battery asymptotics of
the $\LOLP,$ which sheds light on the limits of reliability achievable
in a given setting, as well as the storage capacity required to
achieve a certain (small) $\LOLP$ target.



Finally, we define a quantity that plays a key role in the large
battery asymptotics, namely the \emph{drift} associated with the
supplyside and demandside uncertainty. The drift $\Delta$ is defined
as the steady state average net generation, i.e.,
\begin{equation*}
  \Delta:=\sum_{i\in S} \pi_i r_i.
  \label{eq:drift}
\end{equation*}
Note that $\Delta < 0$ (respectively, $\Delta > 0$) implies that the
time-average generation is less than (respectively, greater than) the
time-average demand.

We conclude this section by considering the special case where the
background CTMC $\{X(t)\}$ has only two states. This simple scenario,
which admits a closed form characterization of the $\LOLP,$ motivates
the general large buffer asymptotics derived in
Section~\ref{sec:bounds}.

\subsection{Two state example}\label{sec:two state example}

Consider the special case $S=\{1,2\}$, where the generation alternates
between two values $0$ and $g > 0$ while the demand takes a constant
value $d \in (0,g)$. Specifically, we set $r_1 = -d,$ $r_2 = g-d.$ In
this case,
\begin{equation*}
   Q=\begin{bmatrix} 
    -a & a \\
    b & -b 
    \end{bmatrix}
\quad
    R=\begin{bmatrix} 
    -d & 0 \\
    0 & g-d 
    \end{bmatrix},
\end{equation*}
where $a,b > 0$ are the state transition rates for the generation process.

In this case, the drift is given by $\Delta=\frac{ag-ad-bd}{a+b},$ and
the $\LOLP$ can be shown to be 
\begin{equation*}
    \LOLP =\frac{-\frac{\Delta}{d}} { 1 - \frac{ag-ad}{bd}\exp\left(\frac{(a+b)\Delta}{(g-d)d}b_{\max} \right)}.
\end{equation*}
It is easy to see that $\LOLP$ is a strictly decreasing function
of~$b_{\max}.$ However, the limiting behavior of the $\LOLP$ as
$b_{\max} \ra \infty$ depends critically on whether the drift is
positive or negative. When $\Delta<0$, then
\begin{equation*}
  \LOLP
  \xrightarrow{b_{\max}\uparrow\infty} -\frac{\Delta}{d} >
  0 \label{eq:LOLP2states}.
\end{equation*}
This means that the $\LOLP$ remains bounded away from zero for any
finite $b_{max}.$ In other words, when the drift is negative, an
$\LOLP$ less than $-\frac{\Delta}{d}$ is simply unattainable no matter
how large the battery capacity. This is consistent with
Theorem~\ref{thm:LOLP_lower_bound} in Section~\ref{sec:bounds}, which
establishes a positive lower bound on the $\LOLP$ for any battery size
$b_{\max}$ when the drift is negative.

On the other hand, when $\Delta > 0,$
\begin{align}
  \LOLP & \sim A e^{ -\lambda b_{\max}},
  \nonumber
\end{align}
where $A=\frac{b}{a}\frac{(a+b)}{(g-d)}\Delta$ and $\lambda =
\left(\frac{a}{d}-\frac{b}{g-d}\right) > 0.$\footnote{We use $f(t)
  \sim g(t)$ to mean that $\lim_{t \ra \infty} \frac{f(t)}{g(t)} =
  1.$} This implies that when the drift is positive, the $\LOLP$
decays exponentially with the battery size, implying that an
arbitrarily small $\LOLP$ target is achievable with a large enough
battery. Moreover, we note that the decay rate $\lambda$ is in fact
the positive eigenvalue of $R^{-1}Q\t$. This is consistent with
Theorem \ref{thm:LOLP_upper_bound} in Section~\ref{sec:bounds}, which
establishes an exponential decay (in the battery size) of the $\LOLP$
when the drift is positive.

\section{Large Battery Approximations}
\label{sec:bounds}

In this section, we analyse the behavior of the $\LOLP$ as the battery
size $b_{\max}$ scales to infinity. Our results shed light on the
feasibility of meeting reliability targets, and also guide the sizing
of the battery required to meet a given reliability target.


As suggested by the two-state example in Section~\ref{sec:model}, the
asymptotic behavior of the $\LOLP$ as $b_{\max} \ra \infty$ depends on
whether the drift is positive or negative. Accordingly, we consider
these cases separately.

\subsection{Negative drift: Asymptotic $\LOLP$ lower bound}\label{sec:neg_drift_lower_bound}

We now consider the case $\Delta < 0,$ i.e., the time-average
generation is less than the time-average demand. One would therefore
expect that $\LOLP$ cannot be made arbitrarily small in this
case. This is proved formally in Theorem~\ref{thm:LOLP_lower_bound},
which also provides a lower bound on the $\LOLP$ that is achievable
with any finite battery size.

Let $\underline{r}:= \min\{r_i,\ i=1,\hdots, |S|\}$. Note that $\underline{r} <
0,$ since we assume that $S_-$ is non-empty; $|\underline{r}|$ is simply
the maximum rate of discharge of the battery.
\begin{theorem}
  \label{thm:LOLP_lower_bound}
  If $\Delta <0,$ then $\LOLP > \frac{-\Delta}{-\underline{r}}$ for any
  value of $b_{\max}.$ Moreover, $$\lim_{b_{\max}\to \infty}\LOLP \geq
  \frac{-\Delta}{-\underline{r}},$$ with equality if $|S_-|=1.$
\end{theorem}
Theorem \ref{thm:LOLP_lower_bound} is a consequence of the law of
large numbers for Markov chains. It states that when the steady state
average demand exceeds the steady state average generation, then an
$\LOLP$ less than or equal to $\nicefrac{-\Delta}{-\underline{r}}$ is
unattainable no matter how large a battery we deploy. Moreover, this
bound is loose in general; it is tight when the background CTMC has
only a single state of discharge. The proof of
Theorem~\ref{thm:LOLP_lower_bound} can be found in
Appendix~\ref{app:lower_bound}.


\textit{Connection with the two state example:} In the two-state
example considered in Section~\ref{sec:model}, note that $\underline{r} =
-d$ and $|S_-| = 1.$ In this example, when $\Delta < 0,$ recall that
indeed, $\LOLP > \frac{-\Delta}{d},$ with $\lim_{b_{\max} \ra \infty}
\LOLP = \frac{-\Delta}{d}.$

\ignore{
Tightness of lower bound is clear from the two state
example. With average drift $\Delta=\frac{ag-ad-bd}{a+b}$ and $
\underline{r}=-d$, $\LOLP$ lower bound obtained in Theorem
\ref{thm:LOLP_lower_bound} is
\begin{equation*}
    \LOLP_{lower\ bound}= \frac{ad+bd-ag}{(a+b)d}=-\frac{\Delta}{d}.
\end{equation*}
This lower bound value is same as the asymptotic $\LOLP$ obtained in
\eqref{eq:LOLP2states} for large battery size ($b_{\max}\uparrow
\infty$). Thus $\LOLP$ lower bound specified by
\Cref{thm:LOLP_lower_bound} becomes tight in the two state example as
$b_{\max}\uparrow \infty$ and $\LOLP$ cannot be reduced further no
matter how large battery size we choose.}

\subsection{Positive drift: $LOLP$ asymptotics}

We now consider the case $\Delta > 0,$ i.e., the time-average
generation exceeds the time-average demand. In this case, one might
expect that it is possible, with a large enough battery, to store the
excess generation when the instantaneous generation exceeds demand,
and to use this stored energy to almost always fulfil the deficit when
the instantaneous generation drops below the
demand. Theorem~\ref{thm:LOLP_upper_bound} shows that this is indeed
the case, and that the $\LOLP$ decays exponentially with the battery
size (when the drift is positive). Moreover,
Theorem~\ref{thm:LOLP_upper_bound} provides two characterizations of
this exponential rate of decay: one from large deviations theory, and
the other as the smallest positive eigenvalue of $R\inv Q\t.$

We now introduce some preliminaries required to state our large
deviations decay rate characterization (Theorem~\ref{thm:ldp}). We
uniformize (see \cite{Kijima1997} for background on uniformization of
CTMCs) the background Markov process $X(\cdot)$ such that the outgoing
rate out of each state equals $$q > \max_{1 \leq i \leq |S|}
-Q_{i,i};$$ recall that $Q$ denotes the rate matrix corresponding to
$X(\cdot).$\footnote{The $i,j$th entry of a matrix $M$ is denoted as
  $M_{i,j}.$} Let $\{Y_k\}$ denote the sequence of intervals between
state transitions in this uniformized chain; note that~$\{Y_k\}$ is an
i.i.d. sequence of $\mathrm{Exp}(q)$ random
variables.\footnote{$\mathrm{Exp}(q)$ refers to the exponential
  distribution with mean $1/q$.}  Let~$\{Z_k\}$ denote the embedded
Markov chain 
corresponding to the (uniformized) Markov process $X(\cdot).$
$\{Z_k\}$ is now a time homogeneous discrete time Markov chain (DTMC);
we denote by $P$ the transition probability matrix corresponding to
this DTMC. We make the following observations.
\begin{enumerate}
\item The DTMC $\{Z_k\}$ is independent of the sequence $\{Y_k\}.$
\item $\pi,$ which denotes the invariant distribution corresponding to
  the background Markov process $X(\cdot),$ is also the invariant
  distribution corresponding to the embedded DTMC $\{Z_k\}.$
\end{enumerate}
Define $U_0 = 0,$ $$U_k := \sum_{j=1}^k -r({Z_j}) Y_j \quad (k \geq
1).$$ The process $\{\nicefrac{U_k}{k}\}$ satisfies a large deviations
principle (this follows from the Gartner-Ellis conditions
\cite{Dembo1998}; see Appendix~\ref{app:positive_drift}), with a rate
function that is defined in terms of the following function.
\begin{equation*}
  \label{eq:Lambda}
 \Lambda(\theta):=\lim_{k \ra \infty} \frac{\log \Exp{e^{\theta
       U_k}}}{k}.
\end{equation*}
That $\Lambda(\cdot)$ is well defined, i.e., the limit in the above
definition exists as an extended real number for all $\theta,$ is
shown in Lemma~\ref{lemma:Lambda_PF} in
Appendix~\ref{app:positive_drift}. 

\begin{theorem}
  \label{thm:LOLP_upper_bound}
  If $\Delta >0,$ then
  \label{thm:ldp}
  $$\lim_{b_{\max} \ra \infty} \frac{\log \LOLP}{b_{\max}} = -
  \lambda,$$ where
  \begin{equation}
  \label{eq:decay_rate}
  \lambda := \sup\{\theta > 0\ :\
  \Lambda(\theta) < 0\}  \in (0,\infty).  
\end{equation}
Moreover, $\lambda$ also equals the smallest positive eigenvalue of
$R\inv Q\t.$
\end{theorem}

Theorem~\ref{thm:LOLP_upper_bound} states that the $\LOLP$ decays
exponentially with respect to the battery size with decay rate
$\lambda.$ 
This ensures that any arbitrarily small $\LOLP$ target be achieved
with a suitably large battery. Additionally,
Theorem~\ref{thm:LOLP_upper_bound} provides an explicit
characterization of this exponential rate, which can in turn be used
to estimate of the battery size required in order to meet a given
(small) $\LOLP$ target; we address battery sizing in detail as part of
our case study (see Section~\ref{case study}).

\textit{Connection with the two state example:} Recall that in the two
state example considered in Section~\ref{sec:model}, we saw that when
$\Delta > 0,$ $\LOLP \sim A e^{-\lambda_c b_{\max}},$ where
$\lambda_c$ is the only positive eigenvalue of $R^{-1}Q\t.$

\subsection*{Proof of Theorem~\ref{thm:LOLP_upper_bound}}

We analyse the large buffer asymptotics of the $\LOLP$ via
the \emph{reversed} system \cite{mitra1988stochastic}, which is obtained by
interchanging the role of generation and demand. Thus, $Q^r = Q,$ and
$R^r = -R,$ where we use the superscript $r$ to represent quantities
in the reversed system. Moreover, $\Delta^r = -\Delta.$ Since the
original system is associated with a positive drift ($\Delta > 0$),
the reversed system is associated with negative drift ($\Delta^r <
0$).

The $\LOLP$ associated with the original system is captured in the
reversed system as follows.
\begin{equation}
\label{eq:lolp_reversed_system}
\LOLP = \prob{b = 0} \stackrel{(\star)}= \prob{b^r = b_{\max}} \stackrel{(\star \star)}\leq \prob{b^r_{\infty} \geq b_{\max}}.
\end{equation}
Here, $b^r_{\infty}$ denotes the stationary buffer occupancy in the
reversed system with an infinite buffer. Note that $\prob{b^r_{\infty}
  \geq b_{\max}}$ is well defined since $\Delta^r < 0.$ The equality
$(\star)$ in \eqref{eq:lolp_reversed_system}, which states that the
long run fraction of time the battery is empty in the original system
equals the long run fraction of time the battery is full in the
reversed system, was first shown in \cite{mitra1988stochastic}. The
inequality $(\star\star)$ follows from a straighforward sample path
argument; by coupling the background process between the finite and
infinite buffer systems, taking $b^r(0) = b^r_{\infty}(0),$ it is not
hard to show that $b^r(t) \leq b^r_{\infty}(t)$ for all $t > 0.$

The asymptotics of $\prob{b^r_{\infty} \geq b_{\max}}$ have been
established via a direct analysis of the invariant distribution of the
process $(X(t),b^r_{\infty}(t))$ in \cite{mitra1988stochastic}:
\begin{equation}
\label{eq:MitraAsymptotics}
 \prob{b^r_{\infty} \geq b_{\max}} \sim A e^{-\lambda b_{\max}},
\end{equation}
where $A > 0$ and $\lambda$ is the smallest positive eigenvalue of
$R\inv Q\t.$

In light of \eqref{eq:lolp_reversed_system}
and \eqref{eq:MitraAsymptotics}, it suffices to prove the following
statements.
\begin{lemma} Let $b^r_\infty$ be the steady state battery occupancy
  level of the infinite battery $(b_{\max}=\infty)$ of the reversed
  system described above.
\label{lemma:ldp_inf_buffer}
$$\lim_{b_{\max} \ra \infty} \frac{\log \prob{b^r_{\infty} \geq
b_{\max}}}{b_{\max}} = -\sup\{\theta > 0\ :\ \Lambda(\theta) < 0\}.$$
\end{lemma}

\begin{lemma} Let $b^r$ be the steady state battery occupancy level of
  the finite battery of the reversed system described above.
\label{lemma:ldp_finite_buffer}
$$\lim_{b_{\max} \ra \infty} \frac{\log \prob{b^r =
    b_{\max}}}{b_{\max}} = -\sup\{\theta > 0\ :\ \Lambda(\theta) <
0\}.$$
\end{lemma}
Lemmas~\ref{lemma:ldp_inf_buffer} and~\ref{lemma:ldp_finite_buffer}
are proved in Appendix~\ref{app:inf_buffer_ld} and
Appendix~\ref{app:finite_buffer_ld} respectively, using large
deviations arguments.

\section{Case Study}
\label{case study}

In this section, we demonstrate the applicability of the results
presented in Section~\ref{sec:bounds} in practice. We fit a Markov
model to a real-world trace of wind power generation, allowing us to
validate the predictions from our analytical results against empirical
observations. Further, we address the question of battery sizing in
order to meet a given reliability target.


\subsection{Data collection}
We collected time series data corresponding to three years of wind
power generation (December 2014 to December 2017) within the
jurisdiction of the Bonneville Power Administration (BPA) (see
\cite{url_wind_data}). The data samples are five minutes apart, and
range from 0 to 4500 MW.

As expected, the data is highly non-stationary in nature, exhibiting
diurnal as well as seasonal variations. Since our Markov modeling is
best suited to stationary data, we extracted the samples corresponding
to the months of February and March from 9PM to 3AM for fitting a
Markov model; this restricted dataset is henceforth referred to as the
`stationary wind data'. For comparison, we also fit a Markov model to
the entire (highly non-stationary) time series.

\ignore{time series data. It is the total wind power collected from
  wind farms in Bonneville Power Administration (BPA) controlled area
  in USA for the period of December 2014 to December 2017. The data
  range is from 0 to 4500 kilowatt (MW). Every successive data sample
  is five minutes apart in time. The data can be found at
  \cite{url_wind_data}.  This raw wind data is inherently
  non-stationary in nature. We verified this by taking time series
  average over a month for a particular time of a day and found that
  the average changes across seasons and months. The data samples from
  9PM to 3AM in February and March are found to be stationary. We
  separated out this stationary wind data and worked on these two data
  sets, referring to them as \lq{}non-stationary wind data' and
  \lq{}stationary wind data'.}

\subsection{Data processing and Markov modeling}

We now describe how we fit a Markov model to the above wind
data.\footnote{This has been attempted before by several authors,
  including \cite{brokish2009pitfalls, nfaoui2004stochastic,
    shamshad2005first, papaefthymiou2008mcmc}. However, these prior
  works evalaute the `fit' quality of their Markov models using the
  mean and auto-correlation function. In contrast, we match the
  reliability implied by the Markov model against the empirical
  reliability, which is a more direct indicator of the usefulness of
  the model.}  We first quantize the data into $N = 20$ bins, the bin
edges being (in MW): [0, 60, 120, 180, 240,300, 450, 600, 900, 1200,
  1500, 1800, 2100, 2400, 2700, 3000, 3300, 3600, 3900, 4200,
  4500]. This non-uniform binning is done to ensure a roughly even
distribution of samples across bins. The $N$ bins constitute the state
space for our Markov model.

\ignore{ These 20 bins forms a state space over which the wind process
  modelled as Markov process will evolve. We found that there are more
  samples in the lower part of the data range (0 to 300MW). Therefore
  such non-uniform binning is deliberately used so that an
  approximately flat histogram across bin sizes can be obtained.}

Given this state space, we obtain the empirical transition probability
matrix $T$ as follows:
\begin{equation*}
 T[i,j]= \frac{\text{\# transitions occurring from bin $i$ to bin
     $j$}}{\text{total \# transitions occurring out of bin $i$}}
\end{equation*}
$T$ is the maximum likelihood estimator of the transition probability
matrix corresponding to a discrete-time Markov chain (DTMC) model for
the wind power sampled at $\tau = 5\ \mathrm{min}$ intervals. To
obtain a continuous-time Markov chain (CTMC) description, we note that
the transition rate matrix $Q$ of the CTMC is related to $T$ as
follows: $T=e^{Q\tau}$. Using the first-order Taylor series
approximation for small $\tau$, we get $e^{Q\tau}\approx I+Q\tau$,
where $I$ is the identity matrix.\footnote{This Taylor approximation
  is valid to long as $\tau$ is smaller than the typical transition
  times of the CTMC.}  Accordingly, we set $Q=(T-I)/\tau$. This $Q$
matrix defines a CTMC description of the wind power data.

To define the net generation corresponding to each state, we assume a
constant demand $d$ over time. Thus, the net generation rate $r_i$
corresponding to bin~$i$ equals $g_i - d,$ where~$g_i$ denotes the
bin-center corresponding to bin~$i.$ Note that we can control the
drift $\Delta$ by varying $d.$

\subsection{Evaluating the goodness of fit}

We now evaluate the quality of our Markov models by comparing the
$\LOLP$ implied by these models with the empirical $\LOLP$ implied by
the data. This also allows us to demonstrate the applicability of the
conclusions of Theorems~\ref{thm:LOLP_lower_bound}
and~\ref{thm:LOLP_upper_bound} in practice.  In
Figures~\ref{fig:neg_drift_LOLP_plots}
and~\ref{fig:pos_drift_LOLP_plots}, we plot the $\LOLP$ as a function
of the battery size $b_{\max}$ setting $d =$ 1800~MW ($\Delta < 0$)
and $d =$ 1200~MW ($\Delta > 0$), respectively. We do this for the
`stationary wind data' as well as the entire time
series. Specifically, we plot the following quantities:
\begin{itemize}
  \item \emph{Simulated (cont. time) $\LOLP$:} This is the $\LOLP$
    computed by simulating the CTMC model for wind power generation
    obtained from the data.
  \item \emph{Simulated (discrete time) $\LOLP$:} This is the $\LOLP$
    computed by simulating the DTMC model for wind power generation
    obtained from the data, taking the generation to be constant over
    5 minute intervals.
  \item \emph{Empirical $\LOLP$:} This is the $\LOLP$ computed by
    simulating the battery evolution using the wind power generation
    trace, again assuming the generation to be constant over 5 minute
    intervals.
\end{itemize}

\ignore{
 A suitable demand $d$ is chosen to create a negative or positive
 drift scenario.  The DTMC is simulated with transition probability
 matrix $T$ for a sufficiently long time.  Each transition in the DTMC
 occurs after every five minutes.  While simulating a CTMC, an
 exponential random variable is used to find the state transition
 instances. At each state transition instance, the next random state
 of the CTMC is found using the rate matrix $Q$.  For a fixed battery
 size $b_{\max}$ and for a chosen wind process (DTMC, CTMC or the
 given wind data), $\LOLP$ is calculated as the long run fraction of
 time for which the battery is empty. We repeat the above $\LOLP$
 calculations for different battery sizes to find how $\LOLP$ is
 structurally dependent on the battery size. This entire analysis is
 done for both sets of data: \lq{}non-stationary wind data' and
 \lq{}stationary wind data'. Results are shown in
 \Cref{fig:neg_drift_LOLP_plots}, \Cref{fig:pos_drift_LOLP_plots},
 \Cref{fig:logLOLP}.
}

\begin{figure}[t]
     \centering \includegraphics[scale=0.5,trim={.5cm 0 0
         0}]{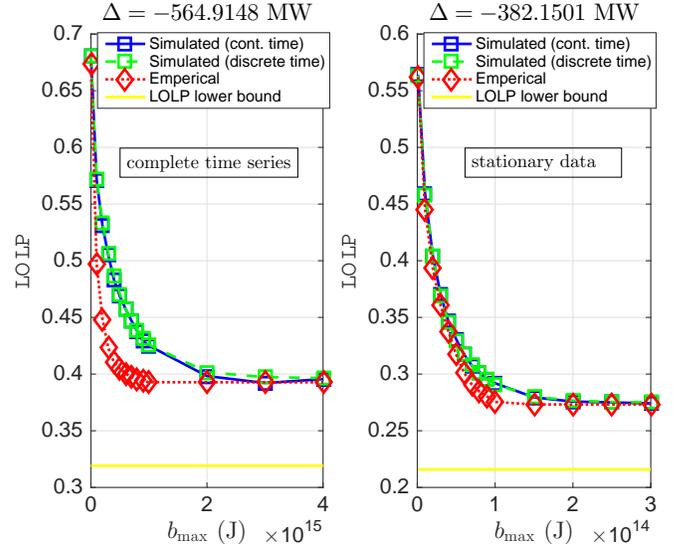}
       \caption{$\LOLP$ vs battery size for $d =$ 1800 MW ($\Delta <
         0$)}
       \label{fig:neg_drift_LOLP_plots}
 \end{figure}

Note that in all the plots, the simlulated $\LOLP$ from our CTMC model
closely matches the simulated $\LOLP$ from the DTMC model. This
essentially validates our first order Taylor approximation for fitting
the transition rate matrix $Q$ from the empirical transition
probability matrix $T.$ Moreover, we note that the simulated $\LOLP$
from the Markov models more closely matches the empirical $\LOLP$ for
the stationary wind data than for the entire time series. This
suggests that the Markov models are a better fit on the stationary
data than on the complete, highly non-stationary time series. In
practice, this means we should fit different Markov models to capture
wind variability in different parts of the day in each season.

Focusing specifically on Figure~\ref{fig:neg_drift_LOLP_plots}, which
corresponds to the negative drift scenario, we make the following
observations.
\begin{itemize}
\item The empirical as well as simulated $\LOLP$ converges, as
  $b_{\max}$ becomes large, to a value which is lower bounded by the
  bound specified in \Cref{thm:LOLP_lower_bound}.
\item The empirical $\LOLP$ is less than the $\LOLP$ implied by the
  Markov models. In other words, our models tend to overestimate the
  $\LOLP.$
\item The $\LOLP$ corresponding to a given battery size\footnote{We
    plot battery size in SI units (Joules). However, the engineering
    practice is to measure battery capacity in kiloWatt-hour (kWh),
    where 1 kWh $=3.6\times 10^6$ J. 
}
  is greater for the entire time series as compared to the stationary
  data, suggesting that the former dataset is more `variable' than the
  latter.
\end{itemize}

Focusing next on Figure~\ref{fig:pos_drift_LOLP_plots}, which
corresponds to the positive drift scenario, we note that the $\LOLP$
decays to zero as $b_{\max}$ becomes large, consistent with
Theorem~\ref{thm:LOLP_upper_bound}. Moreover, we see that the Markov
models tend to overestimate the $\LOLP$ (as before). To illustrate the
exponential decay of $\LOLP$ with battery size clearly, we plot the
simulated $\LOLP$ from the CTMC model on a log-linear scale in
\Cref{fig:logLOLP}. Note that the plot looks asymptotically linear
(establishing the exponential decay), with a slope that closely
matches the decay rate from Theorem~\ref{thm:LOLP_upper_bound}. 

\begin{figure}[t]
  \centering
  \includegraphics[scale=0.5,trim={1cm 0 0 0}]{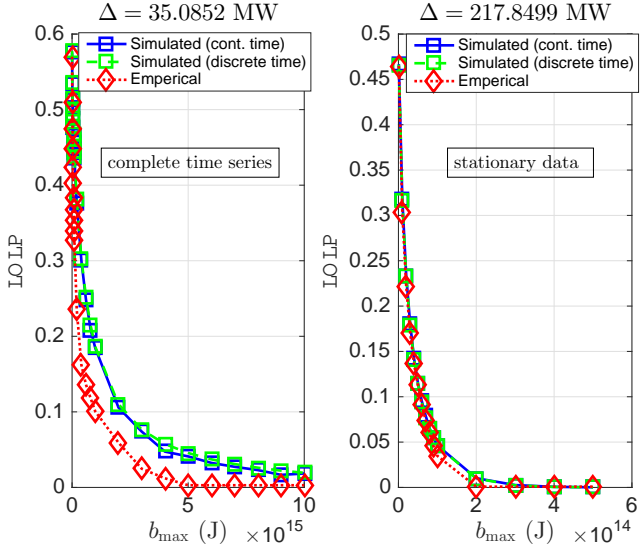}
  \caption{$\LOLP$ vs battery size for $d =$ 1200 MW ($\Delta > 0$)}
  \label{fig:pos_drift_LOLP_plots}
\end{figure}

\subsection{Battery sizing}

The above results support our claim 
that when $\Delta > 0,$ the $\LOLP$ decays exponentially with battery
size with a decay rate equal to $\lambda.$ In other words, when
$b_{\max}$ is large, the $\LOLP$ may be approximated as
\begin{equation}
  \label{eq:lolp_approx}
  \LOLP \approx c  e^{-\lambda b_{\max}}. 
\end{equation}
This further implies that the battery size required to maintain the
$\LOLP$ at $\delta$ is given by 
\begin{equation*}
  b_{\max} \approx \frac{\log(c)}{\lambda} + \frac{\log(1/\delta)}{\lambda}. 
\end{equation*}
Since the pre-factor $c$ in \eqref{eq:lolp_approx} is unknown here, a
natural approximation would be to estimate the battery size required
as
\begin{equation}
  \label{eq:battery_approx}
  b_{\max} \approx \frac{\log(1/\delta)}{\lambda}. 
\end{equation}
Clearly, we would expect the above estimate to be accurate upto an
additive offset. Moreover, we would expect that the error of our
estimate would be small in relative terms for small~$\delta.$

To validate \eqref{eq:battery_approx}, consider the CTMC model for the
stationary wind data, with $d = 1200$ MW. For this model, we compare
the minimum storage size required to bring the simulated $\LOLP$ below
$\delta$ with the estimate \eqref{eq:battery_approx}; see the left
panel of Figure~\ref{fig:BatterySizing2}. Notice the constant offset
between the two curves, as predicted. However, we note the (unknown)
offset results in a roughly 40\% error in battery size requirement
when $\delta = 10^{-3}.$ For lower values of $\delta,$ the relative
error would of course be smaller. This means that for moderate values
of reliability target $\delta,$ the estimate \eqref{eq:battery_approx}
can be used to make ballpark estimates of the storage size required.

However, \eqref{eq:lolp_approx} can also be used for \emph{relative}
storage sizing as follows: Note that \eqref{eq:lolp_approx} suggests
that shrinking the $\LOLP$ be a factor of $\epsilon$ would require an
increase in battery size of $\frac{\log(1/\epsilon)}{\lambda}.$ To
validate this approximation, we consider the following baseline
scenario. Setting $d = 1200$ MW with the stationary wind data, and
$b_{\max} = 0.25\times 10^{12}$ J, the simulated $\LOLP$ equals $L =
0.018.$ In the right panel of \Cref{fig:BatterySizing2}, we plot the
\emph{additional} battery size required to make the $\LOLP$
$L/\epsilon$ versus $\epsilon,$ using the above approximation, as well
as by simulating the CTMC model. Note that the approximation is
remarkably accurate, even for moderate values of $\LOLP.$

This shows that \eqref{eq:lolp_approx} is an accurate description of
the $\LOLP$ as $b_{\max}$ becomes large, and can be used in practice
to guide battery sizing decisions.

 \begin{figure}[t]
     \centering
     \includegraphics[scale=0.5,trim={1cm 0 0 0}]{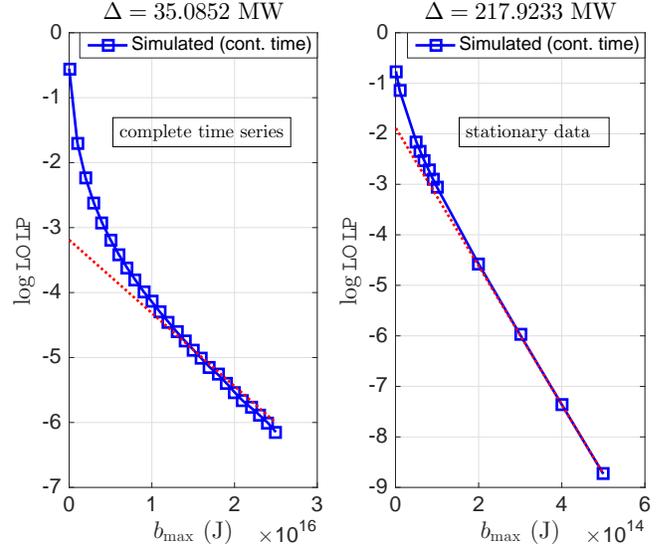}
     \caption{Simulated $\LOLP$ vs battery size plot on $\log$-linear
       scale for for $d =$ 1200 MW (positive drift). The dotted red
       line has slope $-\lambda.$}
     \label{fig:logLOLP}
 \end{figure}

 \begin{figure}
 \includegraphics[scale=0.5]{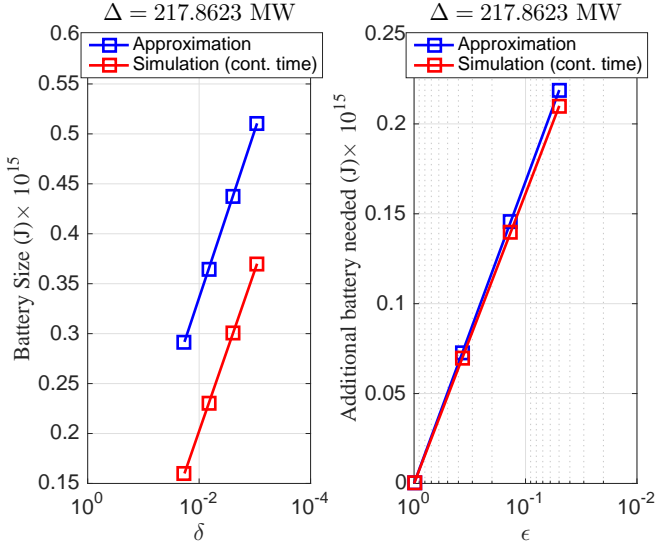}
 \caption{Validation of battery sizing approximations}
 \label{fig:BatterySizing2}
 \end{figure}

\section{Concluding Remarks}

In this paper, we developed an analytical framework for characterizing
the reliability of a renewable generator bundled with a battery. We
analysed how the reliability, captured by the $\LOLP,$ scales as the
battery size increases. Our results highlight the achievable
limits of reliability, and provide useful guidelines for sizing
storage in practice.

While we have used $\LOLP$ as the reliability metric throughout this
paper, it should be noted that our conclusions extend readily to
another related metric, i.e., \emph{lost load rate} ($\mathsf{LLR}$),
which is defined as the long run rate of unserved load,
i.e., $$\mathsf{LLR} = \sum_{i \in S_{-}} F_i(0) (-r_i).$$ Since there
exists positive constants $c_1$ and $c_2$ such that $$c_1 \LOLP \leq
\mathsf{LLR} \leq c_2 \LOLP,$$ our asymptotic characterizations of
$\LOLP$ as $b_{\max} \ra \infty$ extend readily to $\mathsf{LLR}.$
Indeed, when $\Delta < 0,$ $\mathsf{LLR} > -\Delta.$ When $\Delta >
0,$ $\mathsf{LLR}$ decays exponentially with $b_{\max}$ with the same
decay rate $\lambda$ as the one characterized in
Theorem~\ref{thm:LOLP_upper_bound} for $\LOLP.$

This work motivates future research along several directions. 
We believe our formulations are a natural first step to
analyse the economies of scale that would result from sharing of
storage, between renewable generators or electricity
prosumers; we show a result along these lines in~\cite{nair2019statistical}. 
Another direction is performing a similar reliability calculation with for a 
network of generators, taking into account transmission constraints. 
Finally, we note that our work motivates more sound
stochastic modeling of renewable generation, to improve the real-world
applicability of analytical reliability characterizations (as in the
present paper).

\bibliographystyle{IEEEtran}
\bibliography{IEEEabrv,references}

\begin{thebibliography}{10}
\providecommand{\url}[1]{#1}
\csname url@samestyle\endcsname
\providecommand{\newblock}{\relax}
\providecommand{\bibinfo}[2]{#2}
\providecommand{\BIBentrySTDinterwordspacing}{\spaceskip=0pt\relax}
\providecommand{\BIBentryALTinterwordstretchfactor}{4}
\providecommand{\BIBentryALTinterwordspacing}{\spaceskip=\fontdimen2\font plus
\BIBentryALTinterwordstretchfactor\fontdimen3\font minus
  \fontdimen4\font\relax}
\providecommand{\BIBforeignlanguage}[2]{{%
\expandafter\ifx\csname l@#1\endcsname\relax
\typeout{** WARNING: IEEEtran.bst: No hyphenation pattern has been}%
\typeout{** loaded for the language `#1'. Using the pattern for}%
\typeout{** the default language instead.}%
\else
\language=\csname l@#1\endcsname
\fi
#2}}
\providecommand{\BIBdecl}{\relax}
\BIBdecl

\bibitem{deulkar2019sizing}
V.~Deulkar, J.~Nair, and A.~A. Kulkarni, ``{Sizing Storage for Reliable
  Renewable Integration},'' in \emph{IEEE PowerTech}, 2019.

\bibitem{url_energy_storage_association}
E.~S. Association, ``{U.S. Energy Storage Project Pipeline Doubles in 2018,
  Nears 33 GW},''
  \url{http://www.energystorage.org/news/esa-news/us-energy-storage-project-pipeline-doubles-2018-nears-33-gw/},
  2018, [Online; accessed 7-December-2018].

\bibitem{url_US_battery_storage_market_trend}
E.~I. Administration, ``{U.S. Battery Storage Market Trends},''
  \url{https://www.eia.gov/analysis/studies/electricity/batterystorage/}, 2018,
  [Online; accessed 7-December-2018].

\bibitem{anick1982stochastic}
D.~Anick, D.~Mitra, and M.~M. Sondhi, ``Stochastic theory of a data-handling
  system with multiple sources,'' \emph{Bell System Technical Journal},
  vol.~61, no.~8, pp. 1871--1894, 1982.

\bibitem{mitra1988stochastic}
D.~Mitra, ``Stochastic theory of a fluid model of producers and consumers
  coupled by a buffer,'' \emph{Advances in Applied Probability}, vol.~20,
  no.~3, pp. 646--676, 1988.

\bibitem{Kijima1997}
M.~Kijima, \emph{Markov processes for stochastic modeling}.\hskip 1em plus
  0.5em minus 0.4em\relax CRC Press, 1997, vol.~6.

\bibitem{Dembo1998}
A.~Dembo and O.~Zeitouni, \emph{Large deviations techniques and
  applications}.\hskip 1em plus 0.5em minus 0.4em\relax Springer, 1998.

\bibitem{url_wind_data}
B.~P. Administration, ``{Wind flow data},''
  \url{https://transmission.bpa.gov/business/operations/wind/}, 2018, [Online;
  accessed 20-November-2018].

\bibitem{brokish2009pitfalls}
K.~Brokish and J.~Kirtley, ``Pitfalls of modeling wind power using markov
  chains,'' in \emph{Power Systems Conference and Exposition, 2009. PSCE'09.
  IEEE/PES}.\hskip 1em plus 0.5em minus 0.4em\relax IEEE, 2009, pp. 1--6.

\bibitem{nfaoui2004stochastic}
H.~Nfaoui, H.~Essiarab, and A.~Sayigh, ``A stochastic markov chain model for
  simulating wind speed time series at tangiers, morocco,'' \emph{Renewable
  Energy}, vol.~29, no.~8, pp. 1407--1418, 2004.

\bibitem{shamshad2005first}
A.~Shamshad, M.~Bawadi, W.~W. Hussin, T.~Majid, and S.~Sanusi, ``First and
  second order markov chain models for synthetic generation of wind speed time
  series,'' \emph{Energy}, vol.~30, no.~5, pp. 693--708, 2005.

\bibitem{papaefthymiou2008mcmc}
G.~Papaefthymiou and B.~Klockl, ``Mcmc for wind power simulation,'' \emph{IEEE
  transactions on energy conversion}, vol.~23, no.~1, pp. 234--240, 2008.

\bibitem{nair2019statistical}
J.~Nair, A.~A. Kulkarni, and V.~Deulkar, ``Statistical economies of scale in
  battery sharing via large deviations,'' in \emph{IEEE Conference on Decision
  and Control, under review}, 2019.

\bibitem{Wolff1982}
R.~W. Wolff, ``Poisson arrivals see time averages,'' \emph{Operations
  Research}, vol.~30, no.~2, pp. 223--231, 1982.

\bibitem{BigQueues}
A.~J. Ganesh, N.~O'Connell, and D.~J. Wischik, \emph{Big queues}.\hskip 1em
  plus 0.5em minus 0.4em\relax Springer, 2004.

\bibitem{Schwenk1986}
A.~J. Schwenk, ``Tight bounds on the spectral radius of asymmetric nonnegative
  matrices,'' \emph{Linear algebra and its applications}, vol.~75, pp.
  257--265, 1986.

\bibitem{Toomey98}
F.~Toomey, ``Bursty traffic and finite capacity queues,'' \emph{Annals of
  Operations Research}, vol.~79, pp. 45--62, 1998.

\end{thebibliography}

\appendices

\section{Proof of \Cref{thm:LOLP_lower_bound}}
\label{app:lower_bound}

The proof of \Cref{thm:LOLP_lower_bound} is based on energy
conservation of the battery content together with the law of large
numbers applied to the background Markov process.  
\ignore{
We use a following preliminary result in the proof.
\begin{lemma}
  \label{lem:battery_content_relation_infinite_and_finite_capacity}

  Let $b(t)$ and $b_{\infty}(t)$ be finite capacity and infinite
  capacity battery content respectively at any instant $t$. Then,
  $b(t)\leq b_{\infty}(t)\ \forall\ t\geq 0$ on each sample path.
\end{lemma}
The proof of this lemma is fairly simple which is based on the sample path argument and hence, we omit the proof.
}
%
Let $\mathcal{O}_{cum}(t)$ be the total amount of wasted energy due to
battery overflow in the interval $[0,t].$
Similarly, let $\ell_{cum}(t)$ be the total amount of unserved energy
demand in the same interval (during loss of load).
Let $d_{cum}(t)$ be the net demand served over the interval
$[0,t]$ and $g_{cum}(t)$ be total generation over the same
interval. Finally, define
\begin{equation*}
    \ell_{avg}:=\lim_{t \to \infty} \frac{\ell_{cum}(t)}{t}\ \ \ \ \mathcal{O}_{avg}:=\lim_{t \to \infty} \frac{\mathcal{O}_{cum}(t)}{t}.
\end{equation*}
Recall that $\underline{r}:= \min\{r_i,\ i=1,\hdots, |S|\}.$ Let
$\bar{r}:= \max\{r_i,\ i=1,\hdots, |S|\}$. With these notations we
have the following result.

\begin{lemma}
\label{lem:Lemma_neg_drift}
$\ell_{avg}=-\Delta + \mathcal{O}_{avg} > -\Delta$. If the drift is
negative, i.e., $\Delta < 0,$ then $\lim_{b_{\max} \to \infty}
\ell_{avg}=-\Delta$.
\end{lemma}
\begin{IEEEproof}
Applying energy conservation, we get
$$g_{cum}(t) - \mathcal{O}_{cum}(t) = d_{cum}(t)-\ell_{cum}(t) +b(t)$$
\begin{align*}
  \Rightarrow \frac{\ell_{cum}}{t} &= -\frac{(g_{cum}(t)-d_{cum}(t))}{t} +\frac{\mathcal{O}_{cum}(t)}{t} +\frac{b(t)}{t}\\
 \Rightarrow \frac{\ell_{cum}}{t} &= -\frac{1}{t}\int_{0}^{t}r_{X(s)}ds +\frac{\mathcal{O}_{cum}(t)}{t} +\frac{b(t)}{t}\\
 \Rightarrow\lim_{t \to \infty} \frac{\ell_{cum}}{t} &= -\lim_{t \to \infty}\frac{1}{t}\int_{0}^{t}r_{X(s)}ds +\lim_{t \to \infty}\frac{\mathcal{O}_{cum}(t)}{t}\\ & \quad +\lim_{t \to \infty}\frac{b(t)}{t}.
\end{align*}
Since the battery capacity is finite, $\lim\limits_{t \to
  \infty}\frac{b(t)}{t}=0$ and law of large numbers for Markov chains
implies that $\lim_{t \to \infty} \frac{1}{t}\int_{0}^{t}r_{X(s)}ds
=\Delta$. Therefore we get
$$\ell_{avg}=\lim_{t \to \infty}\frac{\ell_{cum}}{t}=-\Delta +\mathcal{O}_{avg}.$$

\ignore{
Also 
$\mathcal{O}_{cum}(t)\leq r_{\max}\times$ \{total time spent in $r_{\max}$ state with full battery content in the interval $[0,t]$\}. 
Dividing by $t$ and taking limit  $t\rightarrow\infty$, we get,
\begin{align}
   \mathcal{O}_{avg}&\leq r_{\max}\mathbb{P}[b=b_{\max}]        \label{eq:CCDF(b_max) finite buffer}\\
   &\leq r_{\max}\mathbb{P}[b_{\infty}\geq b_{\max}]  \label{eq:CCDF(b_max) infinite buffer}
 \end{align}} 

It is easy to see that $$\mathcal{O}_{avg} \leq \bar{r}\
\mathbb{P}[b=b_{\max}] \leq \bar{r} \ \mathbb{P}[b_{\infty}\geq
b_{\max}],$$ where $b_{\infty}$ denotes the stationary buffer
occupancy in an infinite buffer system seeing the same net generation
process.
\ignore{
where $b$ is the stationary battery content corresponding to finite
capacity battery content $b(t)$ whereas $b_{\infty}$ is the stationary
battery content corresponding to the infinite capacity battery content
$b_{\infty}(t)$. Both types of batteries are governed by the same
underlying Markov process $X(t)$. Note that we have used
\Cref{lem:battery_content_relation_infinite_and_finite_capacity} in
\eqref{eq:CCDF(b_max) finite buffer} to get \eqref{eq:CCDF(b_max)
  infinite buffer}.
}
When the drift is negative, as $b_{\max} \ra \infty$,
$\mathbb{P}[b_{\infty}\geq b_{\max}]$ decays exponentially with
$b_{\max}$ (see \cite{mitra1988stochastic}), which implies
that $$\lim_{b_{\max} \to \infty} \ell_{avg}=-\Delta.$$
\end{IEEEproof}

With this result we now prove \Cref{thm:LOLP_lower_bound}.
\begin{IEEEproof}[Proof of \Cref{thm:LOLP_lower_bound}]
It is not hard to see that 
 \begin{equation*}
    \ell_{avg}= \lim_{t \to \infty}\frac{\ell_{cum}(t)}{t} \leq -\underline{r}\LOLP \label{eq:l_avg_upper_bound}.
 \end{equation*}
From \Cref{lem:Lemma_neg_drift}, $\ell_{avg}\geq -\Delta.$
 Therefore, $$\LOLP\geq \frac{-\Delta}{-\underline{r}}.$$

 When the background CTMC has only a single state of discharge, i.e.,
 $|S_-|=1$, then $\ell_{avg}=-\underline{r}\LOLP$.  From
 \Cref{lem:Lemma_neg_drift}, $\lim_{b_{\max} \to \infty}
 \ell_{avg}=-\Delta$ which gives us
$$\lim_{b_{\max}\to \infty}\LOLP =
\frac{-\Delta}{-\underline{r}}.$$
\end{IEEEproof}

\section{Decay rate for $\Delta > 0$}
\label{app:positive_drift}

\subsection{Proof of Lemma~\ref{lemma:ldp_inf_buffer}}
\label{app:inf_buffer_ld}

Let the sequence $\{b^r_{\infty}[n]\}$ denote the buffer occupancy in
the infinite buffer reversed system, sampled at the transition
instants of the (uniformized) background process $X(t).$ We
have
\begin{equation}
\label{disc_Lindley}
b^r_{\infty}[n+1]=(b^r_{\infty}[n]-r_{X_n} Y_n)_+,
\end{equation}
where $(z)_+ = \max(z,0).$ Since the discrete-time process
$\{b^r_{\infty}[\cdot]\}$ is obtained by sampling the continuous-time
process $\{b^r_{\infty}(\cdot)\}$ at the instants of a Poisson process
(of rate $q$), the PASTA property (see \cite{Wolff1982}) implies that
the time averages corresponding to both coincide. Moreover,
\eqref{disc_Lindley} is a Lindley recursion (see \cite{BigQueues})
with negative drift (since $\Delta^r < 0$). The logarithmic
asymptotics of $\prob{b^r_{\infty} \geq B}$ thus follow (see
Theorem~3.1 in \cite{BigQueues}) once we verify that the function
$\Lambda(\cdot)$ is well defined and satisfies the Gartner-Ellis
conditions \cite{Dembo1998}. This is done in
Lemma~\ref{lemma:Lambda_PF} below.


The function $\Lambda(\cdot)$ is characterized as follows.
\begin{lemma}
\label{lemma:Lambda_PF}
\begin{equation*}
\Lambda(\theta) = \left\{\begin{array}{ll}
\log(\rho_{M(\theta)}) & \theta \in \left(\nicefrac{-q}{\bar{r}},\nicefrac{-q}{\underline{r}}\right) \\
\infty & \text{ elsewhere}
\end{array}
\right.,
\end{equation*}
where $\rho_{M(\theta)}$ is the Perron Frobenious eigenvalue
corresponding to the matrix $M(\theta),$ defined as $$M_{l,m}(\theta)
= P_{l,m} \frac{q}{q+\theta r_l},$$
where recall that $P$ is the transition probability matrix of the embedded Markov chain
$\{Z_k\}$. 
 Moreover, $\Lambda(\cdot)$ is
convex and differentiable over
$\left(\nicefrac{-q}{\bar{r}},\nicefrac{-q}{\underline{r}}\right),$
$\Lambda'(0) < 0,$ and 
\begin{align}
  \label{eq:steep}
  \lim_{\theta \ua \nicefrac{-q}{\underline{r}}} \Lambda(\theta) =
  \lim_{\theta \da \nicefrac{-q}{\bar{r}}} \Lambda(\theta) = \infty.
\end{align}
\end{lemma}

\begin{proof}
Throughout the proof, we assume $\theta \in
(\nicefrac{-q}{\bar{r}},\nicefrac{-q}{\underline{r}}).$ It is not hard
to see that $\Exp{e^{\theta U_k}} = \infty$ if $\theta \notin
(\nicefrac{-q}{\bar{r}},\nicefrac{-q}{\underline{r}}).$

Define $v_k(l) = \Exp{e^{\theta U_k}\ |\ X_0 = l},$ $v_k = (v_k(l), l
\in S).$ The vector $v_k$ can be expressed inductively as follows.
\begin{align*}
  v_k(l) &= \sum_{m \in S} P_{l,m} \Exp{e^{-\theta \sum_{j=1}^k r_{X_j}
      Y_j}\ |\ X_2 = m} \Exp{e^{-\theta r_l Y_0}} \\
  &= \sum_{m \in S} P_{l,m} \frac{q}{q+\theta r_l} v_{k-1}(m) \\
  &= \sum_{m \in S} M_{l,m}(\theta) v_{k-1}(m)
\end{align*}
Thus, $v_k = M(\theta) v_{k-1} = M^k(\theta) v_0.$ Denoting the law of
the Markov process at time~$0$ by the row vector
$\pi_0,$ $$\Exp{e^{\theta U_k}} = \pi_0 M^k(\theta) v_0.$$
That $$\lim_{k \ra \infty} \frac{\log \Exp{e^{\theta U_k}}}{k} =
\rho_{M(\theta)}$$ now follows from the Perron Frobenius theorem
\cite[Theorem~3.1.1]{Dembo1998}.

That $\Lambda(\cdot)$ is convex follows from the fact that it is
pointwise limit of convex functions. Its differentiability follows
from the differentiability of Perron Frobenius eigenvalue of a
non-negative matrix with respect to its entries. That $\Lambda'(0) <
0$ follows from Lemma~3.2 in~\cite{BigQueues}.

Finally, to show \eqref{eq:steep}, we use $\rho(M(\theta)) \geq
\rho(G(\theta)),$ where $G(\theta)$ is a symmetric non-negative matrix
defined as $$G_{s,s'}(\theta) =
\sqrt{M_{s,s'}(\theta)M_{s',s}(\theta)};$$ (see Theorem~2 in
\cite{Schwenk1986}). It therefore suffices to show that
$\rho(G(\theta)) \ra \infty$ as $\theta \ua
\nicefrac{-q}{\underline{r}}$ and $\theta \da \nicefrac{-q}{\bar{r}},$
which follows trivially from the observation that
$\mathrm{trace}(G(\theta)) \ra \infty$ along each of the above limits
(note that all diagonal entries of $G$ are positive).
\end{proof}

\subsection{Proof of Lemma~\ref{lemma:ldp_finite_buffer}}
\label{app:finite_buffer_ld}

Lemma~\ref{app:finite_buffer_ld} follows from standard techniques in
large deviations theory that show that the decay rates associated with
$\prob{b^r = b_{\max}}$ and $\prob{b^r_{\infty} \geq b_{\max}}$ as
$b_{\max} \ra \infty$ are the same under very general conditions; see
\cite{Toomey98} and \cite[Section~6.5]{BigQueues}.

\end{document}